\documentclass[11pt]{article}
\newcommand{\PAPER}[1]{}

\usepackage[utf8]{inputenc}

\usepackage{microtype,hyperref}
   \hypersetup{%
      breaklinks,%
      ocgcolorlinks, colorlinks=true,%
      urlcolor=[rgb]{0.25,0.0,0.0},%
      linkcolor=[rgb]{0.5,0.0,0.0},%
      citecolor=[rgb]{0,0.2,0.445},%
      filecolor=[rgb]{0,0,0.4},
      anchorcolor=[rgb]={0.0,0.1,0.2}%
   }

\usepackage[T1]{fontenc}
\usepackage{lmodern}
\usepackage{geometry} 
\usepackage{graphicx}
\usepackage{amsmath,amssymb,amsthm,amsfonts,dsfont,mathtools}
\usepackage{relsize}
\usepackage[margin=1cm]{caption}
\usepackage{wrapfig}
\usepackage{frame,color}
\usepackage{environ,subfig}
\usepackage{xspace}
\usepackage{framed}
\usepackage{comment}
\usepackage{changepage}
\usepackage{enumerate}
\usepackage{algorithm}
\usepackage[noend]{algpseudocode}
\usepackage{algorithmicx}
\usepackage{mathrsfs}
\usepackage{soul}
\usepackage{authblk}
\usepackage{hyperref}
\usepackage[utf8]{inputenc}

\usepackage{url}
\usepackage{breakurl}

\algnewcommand{\Inputs}[1]{%
  \State \textbf{Inputs:}
  \Statex \hspace*{\algorithmicindent}\parbox[t]{.8\linewidth}{\raggedright #1}
}
\algnewcommand{\Initialize}[1]{%
  \State \textbf{Initialize:}
  \Statex \hspace*{\algorithmicindent}\parbox[t]{.8\linewidth}{\raggedright #1}
}
\algnewcommand{\TurnOne}[1]{%
  \State \textbf{Timestep 1:}
  \Statex \hspace*{\algorithmicindent}\parbox[t]{.8\linewidth}{\raggedright #1}
}

\newtheorem{fact}{Fact}[section]
\newtheorem{theorem}[fact]{Theorem}
\newtheorem{lemma}[fact]{Lemma}
\newtheorem{corollary}[fact]{Corollary}
\newtheorem{problem}{Problem} 

\newcommand{\ignore}[1]{}

\newcommand{\ceil}[1]{\left\lceil #1 \right\rceil}
\newcommand{\floor}[1]{\left\lfloor #1 \right\rfloor}

\newcommand{\OO}{\tilde{O}}
\newcommand{\polylog}{\,{\operatorname{polylog}}\,}

\makeatletter
\DeclareRobustCommand\onedot{\futurelet\@let@token\@onedot}
\def\@onedot{\ifx\@let@token.\else.\null\fi\xspace}

\makeatother

\title{More on Change-Making and Related Problems\thanks{Supported in part by NSF Grant CCF-1814026.
A preliminary version of this paper appeared in
ESA 2020~\cite{ChanH20}.
}}

\author[1]{Timothy M. Chan} 
\author[1]{Qizheng He} 
\affil[1]{Department of Computer Science, University of Illinois at Urbana-Champaign}

\begin{document}
\date{}
\maketitle

\begin{abstract}
Given a set of $n$ integer-valued coin types and a target value $t$, the well-known \emph{change-making} problem asks for the minimum number of coins that sum to $t$, assuming an unlimited number of coins in each type.
In the more general
\emph{all-targets} version of the problem, we want the minimum number of coins summing to $j$, for every $j=0,\ldots,t$.
For example, the textbook dynamic programming algorithms can solve the all-targets problem in $O(nt)$ time.
Recently, Chan and He (SOSA'20) described a number of $O(t\,\textrm{polylog}\,t)$-time algorithms for the original (single-target) version of the change-making problem, but not the all-targets version.

In this paper,
we obtain a number of new results on change-making and related problems:

\begin{itemize}
    \item We present a new algorithm for the all-targets  change-making problem with running time $\tilde{O}(t^{4/3})$, improving a previous $\tilde{O}(t^{3/2})$-time algorithm.
    \item We present a very simple $\tilde{O}(u^2+t)$-time algorithm for the all-targets change-making problem, where $u$ denotes the maximum coin value.  The analysis of the algorithm uses a theorem of Erd\H{o}s and Graham (1972) on the Frobenius problem.  This algorithm can be extended to solve the all-capacities version of the \emph{unbounded knapsack} problem
    (for integer item weights bounded by~$u$).
    \item For the original (single-target) coin changing problem, we describe a simple modification of one of Chan and He's algorithms that runs in $\tilde{O}(u)$ time (instead of $\tilde{O}(t)$).
    \item For the original (single-capacity) unbounded knapsack problem,
    we describe a simple algorithm that runs in $\tilde{O}(nu)$ time, improving previous near-$u^2$-time algorithms.
    \item We also observe how one of our ideas implies a new result on the \emph{minimum word break} problem, an optimization version of
    a string problem studied by Bringmann et al.~(FOCS'17),  generalizing
    change-making (which corresponds to the unary special case).
\end{itemize}
\end{abstract}
%
%

\providecommand{\keywords}[1]{\textbf{\textit{Keywords.}} #1}



\section{Introduction}
In the \emph{change-making} problem (also known as \emph{coin changing}), a set of $n$ positive-integer-valued coin types is given, and the cashier wants to use the minimum number of coins to sum to a target value $t$ exactly, where the number of coins in each type can be used an unlimited number of times. This is a well-known textbook problem, which is weakly NP-hard~\cite{lueker1975two}, and standard solutions using dynamic programming~\cite{Wright1975The} have $O(nt)$ running time.

Change-making is closely related to another textbook problem, \emph{subset sum} (the differences are that in subset sum, each item may be used at most once and there is no objective function to minimize).
A series of work in the last few years \cite{bringmann2017near,koiliaris2017faster,jin2018simple,koiliaris2019faster} have given improved algorithms for subset sum, using convolution (FFT\@).  Very recently, at SOSA'20, Chan and He~\cite{chan2020change} revisited the change-making problem and described
a number of  $O(t\polylog t)$-time algorithms, using FFT; their fastest deterministic and randomized algorithms have $O(t\log t\log\log t)$ and  $O(t\log t)$ running time respectively.

\paragraph{All-targets change-making.}
In this paper, we consider a more general, \emph{all-targets} version of the change-making problem: the aim is to compute,
for each target value $j=0,\ldots,t$, the minimum number of coins that can be used to sum to $j$ exactly.  This version of the problem is equally natural.  For instance, the standard $O(nt)$-time dynamic programming algorithms are actually designed to solve this more general version.
Some of the newer subset-sum algorithms~\cite{bringmann2017near,koiliaris2017faster,koiliaris2019faster,jin2018simple} also solved the analogous all-targets version of subset sum, but in contrast, Chan and He's algorithms for change-making do \emph{not} work for the all-targets version.

The best previous result for the all-targets change-making problem that we are aware of was an $\tilde{O}(t^{3/2})$-time\footnote{
The $\tilde{O}$ notation hides polylogarithmic factors.
} algorithm by Karl Bringmann and Tomasz Kociumaka (2019), cited as a personal communication (and briefly sketched) in a very recent paper by Lincoln, Polak, and Vassilevska Williams (ITCS'20)~\cite{lincoln2020monochromatic}.
Lincoln et al.'s paper gave a web of fine-grained reductions connecting a variety of problems, including
a reduction from all-targets change-making to the ``monochromatic convolution'' problem, the latter of which is shown to have near $n^{3/2}$ time complexity iff 3SUM has near quadratic time complexity.
Their work implicitly hints at the possibility that the all-targets change-making problem might have near $t^{3/2}$ complexity as well, but the reduction is in the opposite direction.

Our first result is an $\OO(t^{4/3})$-time algorithm for the all-targets change-making problem, interestingly beating $t^{3/2}$ and placing the problem in a different category than monochromatic convolution and all its surrounding problems.  Our algorithm is conceptually simple, exploiting an easy lemma on a binary special case of $(\min,+)$-convolution (using FFTs).

\paragraph{All-targets change-making in terms of $u$.}
Next, we consider the complexity of the all-targets change-making problem in terms of some other natural parameters besides $n$ and $t$: specifically,

\begin{itemize}
    \item the largest coin value, denoted by $u$;
    \item the sum of the $n$ given coin values, denoted by $\sigma$.
\end{itemize}

Some prior works have analyzed algorithms in terms of $u$ and $\sigma$ for the subset sum problem~\cite{pisinger1999linear,koiliaris2017faster}.
A few recent papers have also analyzed algorithms in terms of $u$ for the \emph{0-1 knapsack} and the \emph{unbounded knapsack} problem~\cite{AxiotisT19,BateniHSS18,EisenbrandW20,JansenR19,Tamir09}.
The unbounded knapsack problem is particularly relevant: given integer weights $w_1,\ldots,w_n$
and profits $p_1,\ldots,p_n$
and capacity value $t$, find nonnegative integers $m_1,\ldots,m_n$ to maximize $\sum_i m_ip_i$ such that $\sum_i m_iw_i\le t$.
Change-making is a special case, for example, by setting $w_i=v_i$ and $p_i=Mv_i-1$ for a sufficiently large $M$.
Improving some previous algorithms~\cite{BateniHSS18,Tamir09},
Axiotis and Tzamos (ICALP'19)~\cite{AxiotisT19}
and Jansen and Rohwedder (ITCS'19)~\cite{JansenR19}
independently described
algorithms\footnote{
We found that an $\OO(u^2)$ algorithm (basically the same as Axiotis and Tzamos') appeared earlier in a commentary on a 2016 programming contest problem by Arthur Nascimento,
solved by Yan Soares Couto;
see Problem~L of
\url{https://www.ime.usp.br/~maratona/assets/seletivas/2016/comentarios.pdf}.
}  for unbounded knapsack  running in $\OO(u^2)$ 
time with $u:=\max_i w_i$
(the time bound can be reduced slightly to $O(u^2/2^{\Theta(\sqrt{\log u})})$ by using known slightly subquadratic algorithms for $(\min,+)$-convolution~\cite{WilliamsAPSP}).
However, these algorithms
do not solve the all-targets or
all-capacities version
(computing the optimal profit for every capacity $j=0,\ldots,t$).\footnote{
Cygan et al.~\cite{CyganMWW19} refered to the all-capacities version as {\sc Unbounded-Knapsack$^+$}; Kunnemann et al.~\cite{kunnemann2017fine} called it the \emph{output-intensive} version.
}

For the all-targets version of change-making, it is not difficult to obtain an $O(u^3+t)$-time algorithm, based on a known observation that when the target is sufficiently large, it is always advantageous to use the largest coin.
We describe a new algorithm that improves the running time to $O(u^2\log u + t)$.
Note that  the algorithm is optimal for large $t\gg u^2\log u$, since the output size for the all-targets problem is $\Omega(t)$.

The new algorithm is remarkably simple---just a slight variation of one of the standard dynamic programming solutions, with a 3-line pseudocode!  (See page~\pageref{alg1}.)  It is easily implementable and does not require FFT\@.  However, the correctness argument is far from obvious, and requires a nice application of a number-theoretic theorem by
Erd\H os and Graham~\cite{erdos1972linear} on the
\emph{Frobenius problem}
(about the smallest target value that cannot be represented by a coin system).  Arguably, algorithms that are simple but nontrivial to analyze are the most interesting kinds of algorithms.

\paragraph{All-capacities unbounded knapsack in terms of $u$.}
Our algorithm can be easily modified to solve the unbounded knapsack problem in the all-capacities version, with the same $O(u^2\log u + t)$ time bound.
This also implies an $O(u^2\log u)$-time algorithm for the single-capacity version, which is a bit simpler than the previous $\OO(u^2)$ algorithms~\cite{AxiotisT19,JansenR19} (in addition to extending it to all-capacities).  For unbounded knapsack, a nearly matching conditional lower bound is known~\cite{CyganMWW19,kunnemann2017fine}: more precisely, if  single-capacity unbounded knapsack could be solved in truly subquadratic time for instances with $t,u=\Theta(n)$, then so could $(\min,+)$-convolution.

\paragraph{In terms of $\sigma$.}
We describe a variant of our algorithm with time bound $\OO((t\sigma)^{2/3}+t)$ for the all-targets change-making or all-capacities unbounded knapsack problem.
Note that if $\sigma\ll t$, this is better than our earlier  $\OO(t^{4/3})$ bound for the all-targets change-making.


\paragraph{Single-target change-making.}
For the single-target (original) change-making problem, we also describe how to improve the running time of one of Chan and He's FFT-based algorithms~\cite{chan2020change} from $\OO(t)$ to $\OO(u)$,
which is faster than applying the previous $\tilde{O}(u^2)$-time algorithms~\cite{AxiotisT19,JansenR19} for single-capacity unbounded knapsack.

\paragraph{Single-capacity unbounded knapsack.}
For the single-capacity (original) unbounded knapsack problem, we also describe a simple algorithm with running time $\OO(nu)$, which (ignoring $n^{o(1)}$ factors) simultaneously improves the standard $O(nt)$-time dynamic programming algorithm and the previous
$\OO(u^2)$-time algorithms~\cite{AxiotisT19,JansenR19} (since $u\le t$ without loss of generality, and $n\le u$ after pruning unnecessary items).
There was a previous $O(nu)$-time algorithm by Pisinger~\cite{pisinger1999linear}  for subset sum, but not for unbounded knapsack.

\paragraph{Minimum word break.}
Finally, we consider a generalization of the problem for strings, known as the \emph{minimum word break} problem:
Given a string $s$ with length $n$ and a set $D$ of strings (a ``dictionary'' of ``words'') with total length $m$, express $s$ as a concatenation of words from $D$, using the smallest number of words, where a word may be used multiple times.
It is easy to see that if the alphabet is unary, then the problem is the same as change-making (the single-target version, with $n$ and $m$ corresponding to $t$ and $\sigma$).
A straightforward dynamic programming algorithm runs in
$\OO(nd+m)$ time, where
$d$ denotes the number of distinct lengths among the words in $D$, by using randomized fingerprints~\cite{backurs2016regular}
(which can be made deterministic~\cite{Xu19}).
Because $m\ge \frac{d(d+1)}{2}$, the bound is $\tilde{O}(n\sqrt{m}+m)$.

The decision version of the problem---deciding whether a solution exists, without minimizing the number of words---was considered by Bringmann, Gr\o nlund, and Larsen~\cite{bringmann2017dichotomy}, who gave an $\OO(nm^{1/3}+m)$-time algorithm, using FFT\@
(improving a previous algorithm by Backurs and Indyk~\cite{backurs2016regular}
with running time $\tilde{O}(nm^{1/2-1/18}+m)$).
Bringmann et al.\ also proved a nearly matching conditional lower bound for combinatorial algorithms, assuming the conjecture that $k$-clique requires near $n^k$ time for combinatorial algorithms.
However, they did not obtain results on the \emph{minimum} word break problem: part of the difficulty is that for the optimization problem, the various convolution operations needed change to $(\min,+)$-convolutions, which appear to be more expensive.

Nevertheless, we note that Bringmann et al.'s algorithm can still be adapted to solve the minimum word break problem.  In fact, the time bound $\OO(nm^{1/3}+m)$ remains the same.
This shows that surprisingly the optimization problem is not harder but has the same fine-grained complexity as the decision problem (at least for combinatorial algorithms, assuming the $k$-clique conjecture).  The only new ingredient in our adaptation of Bringmann et al.'s algorithm is the same lemma on $(\min,+)$-convolutions that we have used in our $\OO(t^{4/3})$ algorithm for change-making.

\section{Preliminaries}
The all-targets version of the change-making problem can be formally defined as follows:
\begin{problem} {\sc (All-Targets Change-Making)}
Given a set $V=\{v_1,\dots,v_n\}$ of $n$ positive integers (coin values) and an integer $t$, for each $j=0,\ldots,t$, find the size of
the smallest multiset $S$ (duplicates allowed) of coin values from $V$ such that $S$ sums to exactly $j$, i.e., find the minimum of $\sum_{i=1}^n m_i$ subject to the constraint that $\sum_{i=1}^n m_iv_i=j$, where $m_i\in \mathds{N}$.
\end{problem}

Besides $n$ (the number of coin values) and $t$ (the maximum target value),
we introduce two more parameters: let $u=\max_{i=1}^n v_i$ denote the maximum coin value, and $\sigma=\sum_{i=1}^n v_i$ denote the sum of input coin values. Simple observation reveals some inequalities relating the parameters: we have $n=O(\sqrt{\sigma})$ (because the distinctness of the $v_i$'s implies $\sigma\geq \frac{n(n+1)}{2}$), $n\leq u$, $u\leq t$ (without loss of generality), and $\sigma\leq nu$.  Note that unlike in the subset sum problem, $t$ may be smaller or larger than $\sigma$.

\paragraph{Boolean convolution.} The Boolean convolution $A\circ B$ of two Boolean arrays $A[0,\dots,t_1]$ and $B[0,\dots,t_2]$ is a Boolean array with $t_1+t_2+1$ elements, where $(A\circ B)[j]=\bigvee_{j'=0}^{t_1}(A[j']\wedge B[j-j'])$ 
(we assume out-of-range values are $0$).

Change-making is closely related with Boolean convolution. For any integer $k$, let $C_V^{(k)}[0,\dots,t]$ denote the Boolean array where
$$ \mbox{$C_V^{(k)}[j]=1$ iff there exist $k$ coins from $V$ with their sum being $j$.}
$$
Then $C_V^{(k)}$ can be obtained from the first $t+1$ elements of $C_V^{(k_1)}\circ C_V^{(k_2)}$, for any $k_1,k_2> 0$ where $k=k_1+k_2$.

The Boolean convolution of two arrays of size $O(t)$
can
be computed in $O(t\log t)$ time by FFT\@.

\paragraph{$(\min,+)$-convolution.} The $(\min,+)$-convolution $A\star B$ of two arrays $A[0,\dots,t_1]$ and $B[0,\dots,t_2]$ is an array with $t_1+t_2+1$ elements, where
$(A\star B)[j]=\min_{j'=0}^{t_1}
(A[j']+B[j-j'])$
(we assume out-of-range values are $\infty$).

Change-making is also related to $(\min,+)$-convolution. For a set $V$ of coin values, let $D_V[0,\dots,t]$ denote the array where $$D_V[j]
= \mbox{the minimum number of coins from $V$ needed to sum to $j$}
$$
(if no solution exists, $D_V[j]=\infty$). Then  $D_{V_1\cup V_2}$ can be obtained from the first $t+1$ elements of $D_{V_1}\star D_{V_2}$.

It has been conjectured by some researchers that $(\min,+)$-convolution cannot be solved in truly subquadratic time (e.g., see~\cite{CyganMWW19,kunnemann2017fine}).
However, the following lemma shows that a subquadratic algorithm is possible
for the special case of $(\min,+)$-convolution where the second array is ``binary'', i.e., all entries of $B$ are in $\{1,\infty\}$.
The lemma (at least the first part) was known before; for example, see a paper by
Kosaraju~\cite{DBLP:conf/focs/Kosaraju89a}, who considered $(\min,\max)$-convolutions, which our special case reduces to.
(A similar trick was also used in the context of matrix multiplication, for computing the $(\min,+)$-product when one of the matrices is binary~\cite{vassilevska2009all,duan2009fast,chan2010more,barr20191}, and for computing the dominance product~\cite{Matousek91}.)


\begin{lemma}\label{lemma:1}
Given two arrays $A[0,\dots,t]$ and $B[0,\dots,t]$ where all entries of $B$ are in $\{1,\infty\}$, we can compute the $(\min,+)$-convolution of $A$ and $B$ in $\tilde{O}(t^{3/2})$ time.
%
%

Furthermore, if we just want $t'$ user-specified entries of the $(\min,+)$-convolution, the time bound
may be reduced to $\tilde{O}(t\sqrt{t'})$.
%
\end{lemma}
\begin{proof}
By sorting and replacing elements by their ranks, we may assume the values of $A$ are in $[t]$, and are distinct (without loss of generality). Divide the range $[t]$ into $\sqrt{t'}$ subintervals of length $t/\sqrt{t'}$.  For each such subinterval $I$, define a Boolean array $A'_I$ with $A'_I[j]=1$ iff $A[j]\in I$, and
define a Boolean array $B'$ with $B'[j]=1$ iff $B[j]\neq\infty$;
compute the Boolean convolution between $A'_I$
and $B'$; this requires $\sqrt{t'}$ FFTs and takes $\OO(t\sqrt{t'})$ time. Then for each index $j$ for which we want to compute the output entry, we can identify which subinterval contains the minimum answer
(namely, the smallest subinterval $I$ such that
$(A'_I\circ B')[j]$ is true) in $O(\sqrt{t'})$ time, so we can do a brute-force search in $O(t/\sqrt{t'})$ time; the total time for $t'$ output entries is
$O(t'\cdot (\sqrt{t'}+t/\sqrt{t'}))=O(t\sqrt{t'})$.
\end{proof}

\section{$\tilde{O}(t^{4/3})$ Algorithm}

\paragraph{Previous algorithm.}
Before presenting the new algorithm, we first give a sketch on the previous $\tilde{O}(t^{3/2})$-time algorithm by Bringmann and Kociumaka  (as mentioned in~\cite{lincoln2020monochromatic}). Let $\ell_0$ be a parameter to be chosen later.
Let $H=\{v_i : v_i > \ell_0\}$
be the set of all \emph{heavy} coin values, and let
$L=\{v_i : v_i \le \ell_0\}$
be the set of all \emph{light} coin values. Because the coin values are distinct, $|L|\le\ell_0$. To sum to any value $j\leq t$, we can use at most $t/\ell_0$ heavy coins. We use Boolean convolution to compute the array $C^{(k)}_H$
from $C^{(k-1)}_H$ for each $k=1,\dots,\lfloor t/\ell_0\rfloor$.  The total time for these $\lfloor t/\ell_0\rfloor$ convolutions is $\OO(t^2/\ell_0)$.
We can thus obtain $D_H[j]$ by
taking the minimum $k\le t/\ell_0$ such that $C_H^{(k)}[j]>0$. To finish, we use the classical dynamic programming algorithm to add the light coins.
Namely, for each $j=1,\dots,t$, we set $D_V[j]=\min\{D_H[j],\min_{v_i\in L}D_V[j-v_i]+1\}$. This
step takes $O(\ell_0 t)$ time.  The overall running time is
\[\OO\left(\frac{t^2}{\ell_0} \,+\, \ell_0 t\right).\]
To balance cost, we choose $\ell_0=\sqrt{t}$ and obtain a time bound of $\OO(t^{3/2})$.

\paragraph{New algorithm.}
To improve the running time, we describe a more efficient way to add the light coins, by using $(\min,+)$-convolution.
As before, we first compute $D_H$
for the heavy coins in $\OO(t^2/\ell_0)$ time.
Initialize $S$ to $H$.


Now, consider a fixed value $\ell\le \ell_0/2$, and consider the subset of light coins $L_\ell=\{v_i: v_i\in (\ell,2\ell]\}$.
In order to add $L_\ell$ to $S$,
we need to compute $D_{S\cup L_\ell}$ from $D_S$.
Naively, one could perform a single $(\min,+)$-convolution of $D_S$ with
$D_{L_\ell}$, but this is expensive, and $D_{L_\ell}$ is not known yet (and is not binary).
A better approach is to
do \emph{multiple} $(\min,+)$-convolutions by dividing the array into smaller blocks of size $O(\ell)$, as follows:

For each $i=0,\dots,t/\ell$, we compute $D_{S\cup L_\ell}[\ell i,\dots,\ell (i+1)]$ by taking a $(\min,+)$-convolution $D'$ of $D_{S\cup L_\ell}[\ell(i-2),\dots,\ell i]$ with a binary array $B[\ell,\dots,2\ell]$ using Lemma \ref{lemma:1}, where $B[j]=1$
if $j\in L_\ell$, and
$B[j]=\infty$ otherwise. Then
$D_{S\cup L_\ell}[\ell i,\dots,\ell (i+1)]$ is
the entry-wise minimum of $D'[\ell i,\dots,\ell(i+1)]$
and $D_S[\ell i,\dots,\ell(i+1)]$, because if the optimal solution (with coin set $S\cup L_\ell$) for a target value in $[\ell i,\ell(i+1)]$ uses a coin in $L_\ell$, then after taking out this coin with value in $(\ell,2\ell]$, the remaining target value is in $[\ell (i-2),\ell i]$. (This explains why we group the coins with roughly the same value in $L_\ell$.)
Each of the above $O(t/\ell)$ $(\min,+)$-convolutions is done to arrays of size $O(\ell)$ (after shifting indices).
Thus, the total running time is $\tilde{O}((t/\ell)\cdot \ell^{3/2})=\tilde{O}(\sqrt{\ell} t)$.

We repeat the above steps for all $\ell$'s that are powers of 2 and smaller than $\ell_0$, until all coin values are added to $S$. This requires $O(\log \ell_0)$ rounds, and the total running time forms a geometric series bounded by $\OO(\sqrt{\ell_0}t)$.
The overall running time is \[\OO\left(\frac{t^2}{\ell_0} \,+\, \sqrt{\ell_0}t\right).
\]
To balance cost, we choose $\ell_0\approx t^{2/3}$ and obtain a time bound of $\tilde{O}(t^{4/3})$.

\begin{theorem}\label{thm:1}
The all-targets change-making problem can be solved in $\tilde{O}(t^{4/3})$ time.
\end{theorem}

\paragraph{Remark.}
If we choose $\ell_0=u$ instead, the heavy coin case can be ignored and we obtain an $\tilde{O}(t\sqrt{u})$-time algorithm,
which is faster for small $u$.  We will give still faster algorithms for small $u$ in the next section.



\section{$O(u^2\log u+t)$ Algorithm}

We now explore more algorithms with running time sensitive to $u$.

\paragraph{Warm-up.}
We first observe that there is a simple
algorithm with $O(u^3+t)$ running time.
We use the following lemma, which is ``folklore'':\footnote{
Bateni et al.~\cite[Lemma~7.2]{BateniHSS18} gave a proof for the (more general) unbounded knapsack problem, using the pigeonhole principle, similar to what we give here
(Eisenbrand and Weismantel~\cite{EisenbrandW20} also proved a similar statement for higher-dimensional unbounded knapsack).
But it was known much earlier: we personally learned of the
pigeonhole proof for coin changing from comments by Bruce Merry in 2006
on a US Olympiad question (\url{https://contest.usaco.org/TESTDATA/DEC06.fewcoins.htm}),
and the same pigeonhole proof
for unbounded knapsack from a Chinese web post in 2016
(\url{https://www.zhihu.com/question/27547892/answer/133582594}).
}

\begin{lemma}\label{lemma:6}
For any target value $j\ge u^2$, any optimal solution to the change-making problem must use the largest coin value $u$.
%
\end{lemma}
\begin{proof}
Suppose that an optimal solution $X$ for a target value $j$ does not use the coin value~$u$.

A simple argument shows that $j < u^3$: If $X$ uses a coin value $v_i$ at least $u$ times,
we can replace $u$ copies of $v_i$ with
$v_i$ copies of $u$, and the number of coins in $X$ would decrease: a contradiction.  Thus, each of the at most $u$ coin values is used fewer than $u$ times, and so the sum of $X$ must be less than $u^3$.

We give a better argument showing $j < u^2$ by using the pigeonhole principle:
Let $\langle x_1,\ldots,x_h\rangle$ be the sequence of coins used in $X$, with duplicates included, in an arbitrary order.
Define the prefix sum $s_i=x_1+\cdots+x_i$.
Suppose $h\ge u$.  By the pigeonhole principle, there must exist $0\le i<j\le h$ with $s_i\equiv s_j\pmod{u}$.  Then the subsequence $x_{i+1},\ldots,x_j$ sums to a number divisible by $u$.  We can replace this subsequence with some number of copies of $u$, and the number of coins in $X$ would decrease (since $u$ is the largest coin value): a contradiction.
Thus $h < u$, and so the sum of $X$ is less than $u^2$.
%
\end{proof}

The above lemma ensures that it is sufficient to compute $D_V[j]$ for all $j < u^2$;
by the naive dynamic programming algorithm, this step takes $O(nu^2)\le O(u^3)$ time.
Afterwards,
for $j=u^2,\ldots,t$, we can simply set $D_V[j]=D_V[j-u]+1$;
this step takes $O(t)$ time.
We thus get the time bound $O(u^3+t)$.

If in the first part we instead use the $\OO(t\sqrt{u})$-time
algorithm in the remark after Theorem~\ref{thm:1} (with $t$ replaced by $u^2$), then
the first part takes
$\OO(u^2\sqrt{u})$ time.
The total time is then reduced to
$O(u^{2.5}\polylog u + t)$.  (This requires FFT, however.)

\paragraph{New algorithm.}
To improve the running time further, we use
number-theoretic results on the \emph{Frobenius problem}, which has received much attention from mathematicians: given $k$ positive integer coin values $v_1>\cdots>v_k$ with $\gcd(v_1,\dots,v_k)=1$, what is the largest number that cannot be represented?  For $k=2$, classical results show that the number is exactly $v_1v_2-v_1-v_2$.  For $k\ge 3$, the problem becomes much more challenging, for which there are no closed-form formulas.  In 1972,
Erd\H{o}s and Graham~\cite{erdos1972linear} proved an upper bound of
$2\floor{\frac{v_1}k}v_2-v_1$, which will be useful in our algorithmic application:


\begin{lemma}\label{lemma:7}
\emph{(Erd\H os--Graham)}
Given integers $v_1>\dots>v_k>0\ (k\ge 2)$ with $\gcd(v_1,\dots,v_k)=1$, any integer greater than $2\floor{\frac{v_1}k}v_2-v_1$ can be expressed as a nonnegative integer linear combination of $v_1,\dots,v_k$.
\end{lemma}

In terms of $u=\max_i v_i$, Erd\H os and Graham's bound is $O(u^2/k)$, which is known to be tight in the worst case, within a constant factor (see~\cite{Dix} for improvements on the constant factor).  For constant~$k$, the bound remains quadratic, as in the 2-coins case. In our algorithmic application, we will consider non-constant $k$---here, the $k$ in the denominator will prove crucial.

First, let us restate the bound  more generally without assuming  $\gcd(v_1,\dots,v_k)=1$:

\begin{corollary}\label{corollary:8}
Given integers $v_1>\dots>v_k>0\ (k\ge 2)$ with $\gcd(v_1,\dots,v_k)=d$, any integer that is greater than $2\floor{\frac{v_1}{dk}}v_2 - v_1$ and is divisible by $d$ can be expressed as a nonnegative integer linear combination of $v_1,\dots,v_k$.
\end{corollary}

\begin{proof}
Apply Lemma \ref{lemma:7} to the numbers $v_1/d,\dots,v_k/d$. The bound becomes \\  $\left(2\floor{\frac{v_1/d}k}v_2/d - v_1/d\right) \cdot d.$
\end{proof}

We use Corollary~\ref{corollary:8} to prove a more refined version of Lemma \ref{lemma:6}, which takes into account the $k$ largest coin values instead of just the largest value:

\begin{lemma}\label{lem:klargest}
Let $v_1,\dots,v_k\le u$ be the $k$ largest input coin values.
For any target value $j\ge 2u^2/k$, any optimal solution to the change-making problem must use at least one coin from $\{v_1,\dots,v_k\}$.
\end{lemma}

\begin{proof}
We may assume $k\ge 2$ (because of Lemma~\ref{lemma:6}).
Let $d = \gcd(v_1,\dots,v_k)$.
Suppose that an optimal solution $X$ for a target value $j$ does not use
any coins from $\{v_1,\dots,v_k\}$.

Consider the sequence of coins used in $X$, with duplicates included,
in an arbitrary order.
Divide the sequence into subsequences $X_1,\dots,X_h$, each of which has sum in $(\frac{2u^2}{dk}-u, \frac{2u^2}{dk}]$, except that the last has sum at most $\frac{2u^2}{dk}-u$. Suppose $h>d$. Define $s_i$ to be the sum of the concatenation of $X_1,\dots,X_i$. By the pigeonhole principle,
there exist $0\le i<j<h$ with $s_i\equiv s_j\pmod{d}$. Then the subsequence formed by concatenating $X_{i+1},\ldots,X_j$ sums to a number divisible by $d$ and greater than $\frac{2u^2}{dk}-u$. By Corollary \ref{corollary:8}, we can replace this subsequence with coins from the set $\{v_1,\dots,v_k\}$, and the number of coins in $X$ would decrease (since $v_1,\dots,v_k$ have larger values): a contradiction. Thus  $h\le d$, and so the sum of $X$ is
less than $d \cdot \frac{2u^2}{dk} = 2u^2/k$.
\end{proof}

Thus, the optimal solution for target value $j$ must use
at least one coin value which is among the $\ceil{2u^2/j}$ largest.
This leads to the following extremely simple algorithm, which
is just a small modification to the standard dynamic programming algorithm (no FFT required):

\begin{algorithm}
\label{alg1}
\begin{algorithmic}[1]
\State Sort $v_1,\dots,v_n$ in decreasing order, and set $D_V[0]=0$.
\For {$j = 1,\dots,t$}
    \State Set $D_V[j] = \min_{1\leq i\leq \ceil{2u^2/j}:\: v_i\le j}  D_V[j-v_i] + 1$.
\EndFor
\end{algorithmic}
\end{algorithm}

The total running time is bounded by a Harmonic series: \[O\left(\sum_{j=1}^{t} \left( \frac{u^2}j + 1\right)\right) \,=\, O(u^2 \log u + t).\]


\begin{theorem}\label{thm:2}
The all-targets change-making problem can be solved in $O(u^2\log u+t)$ time.
\end{theorem}


As a corollary of the above algorithm, we can also obtain an algorithm with running time sensitive to $\sigma$, the total sum of the input coin values:
%
Define the heavy coins $H$ and light coins $L$ as before, with respect to a parameter $\ell_0$ to be chosen later.
We first compute $D_L$ for the light coins by the above algorithm in $\tilde{O}(\ell_0^2 + t)$ time. Then we add the heavy coins by dynamic programming:
$D_V[j] = \min\{ D_L[j], \min_{v_i\in H} D_V[j-v_i] + 1\}$.  Since there are at most $\sigma/\ell_0$ heavy coins, this step takes $O(\sigma/\ell_0\cdot t)$ time.
The overall running time is
\[ \OO\left(\ell_0^2 \,+\, \frac{t\sigma}{\ell_0} \,+\, t\right).
\]
To balance cost, we choose $\ell_0 = (t\sigma)^{1/3}$
and obtain the time bound $\tilde{O}((t\sigma)^{2/3} + t)$.
(Again, no FFT is required.)

\begin{corollary}\label{cor:2}
The all-targets change-making problem can be solved in $\tilde{O}((t\sigma)^{2/3} + t)$ time.
\end{corollary}

\paragraph{Remark.}
The $O(t)$ term can be eliminated in Theorem~\ref{thm:2} (and thus Corollary~\ref{cor:2}) if we are fine with an \emph{implicit} representation of the output (i.e., a structure that allows us to return the answer for any given target in constant time), since by Lemma~\ref{lemma:6}, we can first reduce the target $j$ to below $u^2$ by using some number (i.e., $\max\{\ceil{(j-u^2)/u},0\}$) of copies of the largest coin value $u$.


\section{All-Capacities Unbounded Knapsack}

We note that the algorithm in the preceding section can be extended to solve the all-capacities version of the unbounded knapsack problem, defined as follows:
\begin{problem} {\sc (All-Capacities Unbounded Knapsack)}
Given $n$ items where the $i$-th item has a positive integer
weight $w_i$ and
a positive profit $p_i$, and given an integer $t$, for each $j=0,\ldots,t$, find
the maximum total profit of a multiset of items such that the total weight is at most $j$, i.e.,
find the maximum
of $\sum_{i=1}^n m_ip_i$ subject to the constraint that $\sum_{i=1}^n m_iw_i\le j$, where $m_i\in \mathds{N}$.
\end{problem}

Like before, let $u=\max_{i=1}^n w_i$ and $\sigma=\sum_{i=1}^n w_i$.
We may assume that the weights are distinct (since if there are two items with the same weight, we may remove the one with the smaller profit).

We use the following analog to
Lemma~\ref{lem:klargest}:

\begin{lemma}
Suppose items $1,\dots,k$ have the $k$ largest profit-to-weight ratios. For any capacity value $j\ge 3u^2/k$, any optimal solution to the unbounded knapsack problem must use at least one item from $\{1,\ldots,k\}$.
\end{lemma}
\begin{proof}
Similar to the proof of  Lemma~\ref{lem:klargest}, since replacing a subsequence with items
that have larger profit-to-weight ratios while maintaining the same total weight would increase the total profit.  One difference in the unbounded knapsack problem is that the total weight in the optimal solution may not be exactly $j$.  But it must be at least $j-u$ (otherwise, we could add one more item to get a better solution).  When $j\ge 3u^2/k$, we have $j-u\ge 2u^2/k$.
\end{proof}

The same analysis shows correctness of the following very simple algorithm, which runs in $O(u^2\log u+ t)$ time:
\begin{algorithm}
\label{alg2}
\begin{algorithmic}[1]
\State Sort the items in decreasing order of $p_i/w_i$.
\For {$j = 0,\dots,t$}
    \State Set $D[j] = \max\{0,\ \max_{1\leq i\leq \ceil{3u^2/j}:\: w_i\le j}  (D[j-w_i] + p_i)\}$.
\EndFor
\end{algorithmic}
\end{algorithm}

The $\tilde{O}((t\sigma)^{2/3} + t)$ algorithm can be extended as well.

\begin{corollary}\label{cor:knapsack}
The all-capacities unbounded knapsack problem can be solved in $O(u^2\log u + t)$ or $\tilde{O}((t\sigma)^{2/3} + t)$ time.
\end{corollary}

\paragraph{Remarks.}
As before, the $O(t)$ term can be eliminated with an implicit representation of the output (since by an analog to Lemma~\ref{lemma:6}, we can first reduce the capacity to below $u^2$ by using some number of copies of the item with the largest profit-to-weight ratio).
In particular, for the single-capacity version, we obtain a very
simple $O(u^2\log u)$-time algorithm.

The algorithm works even when the profits are reals but the weights are integers.  Alternatively, a variant of the algorithm works when the weights are reals but the profits are integers:  the same time bound $O(u^2\log u)$ holds but with $u=\max_{i=1}^np_i$.  Here, we recast the problem as minimizing $\sum_{i=1}^n m_iw_i$ subject to the constraint that $\sum_{i=1}^n m_ip_i\ge j$, and modify the algorithm appropriately (applying Erd\H os--Graham to the profits instead of the weights).  From the implicitly represented output, we can determine the answer for any given capacity by predecessor search.

\section{$\tilde{O}(u)$ Algorithm for Single-Target Change-Making}
In this section, we present an $\tilde{O}(u)$-time algorithm for the single-target change-making problem.  It is obtained by modifying the third algorithm in our previous paper~\cite{chan2020change}, which originally ran in $O(t\log^2 t)$ time.
In that algorithm, we first solved the decision problem:
deciding whether we can sum to $t$ using at most $m$ coins for a given value $m$.
By adding 0 to the input set of coin values, ``at most $m$'' can be changed to ``exactly $m$''.


That previous decision algorithm relies on the following partition lemma, which shows the multiset of coins $S$ can be almost evenly partitioned simultaneously in terms of cardinality \emph{and} the total value:
\begin{lemma}[Partition Lemma]
Suppose $S$ is a multiset with $|S|=m$ and has sum $\sigma(S)\triangleq \sum_{s\in S}s=t$. If $m$ is odd, then there exists a partition of $S$ into three parts $S_1$, $S_2$ and a singleton $\{s_0\}$, such that $|S_1|=|S_2|=\frac{m-1}{2}$ and $\sigma(S_1),\sigma(S_2)\leq \frac t2$.

If $m$ is even, then there exists a partition of $S$ into three parts $S_1$, $S_2$ and two elements $\{s_0,s_1\}$, such that $|S_1|=|S_2|=\frac m2-1$, and $\sigma(S_1),\sigma(S_2)\leq \frac t2$.
\end{lemma}
Our previous paper \cite{chan2020change} provided a short proof for the even $m$ case, and here for the benefit of the reader, we restate a self-contained proof for the odd case.
\begin{proof}
Let $s_1,s_2,\ldots,s_m$ be the elements of $S$ in an arbitrary order.
Let $\hat{S}_1$ be the set  of the first $\frac{m-1}{2}$ elements,
and let $\hat{S}_2$  be the set of the last $\frac{m-1}{2}$ elements.
W.l.o.g., assume that  $\sigma(\hat{S}_1)\leq \sigma(\hat{S}_2)$ (for otherwise we can swap these two parts).
If $\sigma(\hat{S}_2)
\leq \frac{t}{2}$ then we can simply set $S_1=\hat{S}_1$, $S_2=\hat{S}_2$, and $s_0$ be the $\frac{m+1}{2}$-th element. Otherwise maintain a sliding window $W$ containing exactly $\frac {m-1}{2}$ consecutive elements of $s_1,\ldots,s_m$.  Initially, $W=\hat{S}_1$ and $\sigma(W)\le \frac t2$.  At the end, $W=\hat{S}_2$ and $\sigma(W)>\frac t2$.
Thus, at some moment in time, we must have $\sigma(W)\le \frac t2$ but $\sigma(W')>\frac t2$, where $W'$ denotes the next window after $W$.
We let $S_1=W$, $s_0$ be the unique element in $W'\setminus W$, and
$S_2=S\setminus (W\cup\{s_0\})$.  Since $S_2\subseteq S\setminus W'$, we have
$\sigma(S_2)<\frac t2$.
\end{proof}

Notice that since the maximum coin value is $u$, we also have $\sigma(S_1),\sigma(S_2)\geq \frac{t}{2}-2u$ (as we take out one or two coins).

The Partition Lemma suggests
a simple recursive algorithm
to compute $C_V^{(m)}[0,\ldots,t]$:
we just take the first $t+1$ entries of
\[
\left\{
\begin{array}{ll}
C_V^{(\frac{m-1}{2})}[0,\ldots,\frac t2]\circ
C_V^{(\frac{m-1}{2})}[0,\ldots,\frac t2]\circ
C_V^{(1)}[0,\ldots,t]
& \mbox{if $m$ is odd},\\
C_V^{(\frac{m}{2}-1)}[0,\ldots,\frac t2]\circ
C_V^{(\frac{m}{2}-1)}[0,\ldots,\frac t2]\circ
C_V^{(1)}[0,\ldots,t]\circ C_V^{(1)}[0,\ldots,t]
& \mbox{if $m$ is even}.
\end{array}
\right.
\]
That was essentially our previous algorithm~\cite{chan2020change}.

We describe a more efficient recursive algorithm to compute a smaller subarray
$C_V^{(m)}[t-4u,\dots,t]$:
we just take the relevant entries of
\[
\left\{
\begin{array}{l}
C_V^{(\frac{m-1}{2})}[\frac{t-4u}{2}-2u,\dots,\frac{t}{2}]\circ C_V^{(\frac{m-1}{2})}[\frac{t-4u}{2}-2u,\dots,\frac{t}{2}]\circ C_V^{(1)}[0,\dots,u]\\
\qquad\qquad\qquad\qquad\qquad\qquad\qquad\qquad\qquad\qquad\qquad\qquad\qquad\qquad\mbox{if $m$ is odd,}\\[2pt]
C_V^{(\frac{m}{2}-1)}[\frac{t-4u}{2}-2u,\dots,\frac{t}{2}]\circ C_V^{(\frac{m}{2}-1)}[\frac{t-4u}{2}-2u,\dots,\frac{t}{2}]\circ C_V^{(1)}[0,\dots,u]\circ C_V^{(1)}[0,\dots,u]\\
\qquad\qquad\qquad\qquad\qquad\qquad\qquad\qquad\qquad\qquad\qquad\qquad\qquad\qquad\mbox{if $m$ is even.}
\end{array}
\right.
\]
Each of the above Boolean convolutions is done to arrays of size $O(u)$ (after shifting indices), and thus takes $O(u\log u)$ time.  The subarrays $C_V^{(\frac{m-1}{2})}[\frac{t-4u}{2}-2u,\dots,\frac{t}{2}]
=C_V^{(\frac{m-1}{2})}[\frac{t}{2}-4u,\dots,\frac{t}{2}]$ and
$C_V^{(\frac{m}{2}-1)}[\frac{t-4u}{2}-2u,\dots,\frac{t}{2}] = C_V^{(\frac{m}{2}-1)}[\frac{t}{2}-4u,\dots,\frac{t}{2}]$ can be computed by recursion.
Thus, the running time satisfies the recurrence
\[ T(m,t) = T(\floor{\tfrac{m-1}{2}},\tfrac t2) + O(u\log u),
\]
which solves to $T(m,t)=O(u\log u\log t)$.

The decision problem can now be solved by inspecting the entry $C_V^{(m)}[t]$.  We can
find the optimal number of coins by binary search with $O(\log t)$ calls to the decision algorithm.
By Lemma \ref{lemma:6}, we can first reduce $t$ to below $u^2$ by repeatedly using the largest coin value. Therefore, the total running time is $O(u\log u\log^2 t)\le O(u\log^3 u)$.

\begin{theorem}\label{thm:single:target}
The single-target change-making problem can be solved in $O(u\log^3 u)$ time.
\end{theorem}

\paragraph{Remarks.} The above algorithm shares some similarity with the $\OO(u^2)$ algorithm by
Axiotis and Tzamos~\cite{AxiotisT19} for unbounded knapsack, which also involves logarithmically many convolutions on subarrays of size $O(u)$, except that they used $(\min,+)$-convolutions and a more naive parititioning that approximately halves $t$, but not $m$.  In contrast, the above Partition Lemma is crucial to our faster algorithm for change-making.

There is also some similarity with an algorithm by Jansen and Rohwedder \cite{JansenR19}, who studied a more general problem of integer programming with a constant number of constraints.
Their algorithm also aims to simultaneously divide the target and the cardinality in half, by using more advanced techniques, namely, ``Steinitz Lemma''.



\section{$\OO(nu)$ Algorithm for Single-Capacity Unbounded Knapsack}

In this section, we revisit the standard (single-capacity) version of the unbounded knapsack problem and present a new $\OO(nu)$-time algorithm (recall that $u=\max_i w_i$).
This algorithm is simple (no FFT needed), and is based on the following combinatorial lemma, which is obtained by another pigeonhole argument:

\begin{lemma}
For the unbounded knapsack problem for a given capacity $j$,
there exists an optimal solution that uses at most $\log j$ different types of items.

In particular, in some optimal solution, there exists an item $i$ that is used at least
$\frac{j}{w_i\log j}$
times.
\end{lemma}
\begin{proof}
Consider an optimal solution that uses the minimum number of types of items.  Let $S$ be the set of items used in this solution, excluding multiplicities.
If $|S|>\log j$, by the pigeonhole principle there must exist two different subsets $S_1$ and $S_2$ of $S$ with the same total weight, multiplicities included  (since there are $2^{|S|}$ subsets and $j$ integers between 0 and $j-1$).
We can replace the items in $S_2\setminus S_1$ with $S_1\setminus S_2$, or vice versa (depending which of the two has smaller total value), and get a new solution that has the same total weight but has larger or equal total value.
And if it has equal total value, the new solution uses a smaller number of types of items (since $S_2\setminus S_1$ and $S_1\setminus S_2$ are nonempty): a contradiction.

Thus, $|S|\le\log j$.  This also implies that some item contributes  at least $\frac{j}{\log j}$ to the total weight.
\end{proof}

Let $b:=\ceil{\log t_0}$, where $t_0$ is an upper bound on the capacity.
Let $D[j]$ be the maximum profit for the unbounded knapsack problem with capacity $j.$ For $t<t_0-u$, letting $\hat{t}:=\lceil(1-\frac1b)t\rceil$, we can compute the subarray $D[t,\ldots,t+u]$ from the subarray $D[\hat{t},\ldots, \hat{t}+u]$ in
$O(nu)$ time, using the following recursive formula for each entry $j$ in $[t,t+u]$:
\[ D[j]\: =\: \max\left\{0, \ \max_{i=1}^n (D[j-w_ix_{ij}] + p_ix_{ij})\right\} \qquad\mbox{where }
x_{ij} := \lceil\tfrac{j-(\hat{t}+u)}{w_i}\rceil.
\]

Note that $x_{ij}\leq \lceil\tfrac{j}{w_ib}\rceil$,
because $j-(\hat{t}+u)\le \tfrac{j}{b}$, i.e.,
$(1-\frac1b)j \le \hat{t}+u$ (which is obvious since
$\hat{t}=\lceil (1-\frac1b)t\rceil$).
Thus, the correctness of the formula follows from the above lemma.
Also note that
$j-w_ix_{ij}\in
[\hat{t},\hat{t}+u]$. The latter subarray $D[\hat{t},\ldots, \hat{t}+u]$ can be computed recursively.

Let $T(t)$ denote the time for computing $D[t,\ldots,t+u]$.  We thus obtain the following recurrence:
\[ T(t) \:=\: T(\lceil(1-\tfrac1{b})t\rceil) \,+\, O(nu).
\]
For the base case, we have $T(b)=O(nu)$
by the standard dynamic programming algorithm (which computes $D[0,\ldots,j]$ in $O(nj)$ time).
The number of levels of recursion is $O(b\log t)$.
So,
$T(t) = O(bnu\log t)=O(nu\log t_0\log t)$.  We can set $t_0=(t+u)^{O(1)}$.
As before, we can initially reduce the capacity $t$ to below $u^2$ by repeatedly using  the item with the largest profit-to-weight ratio.  This yields the following result:

\begin{theorem}\label{thm:knapsack}
The single-capacity unbounded knapsack problem can be solved in $O(nu\log^2 u)$ time.%
\footnote{In the preliminary version of the paper~\cite{ChanH20},
we claimed a slightly weaker $O(nu\log^3 u)$ time bound, due to some small differences in the algorithm.}
\end{theorem}


\section{Minimum Word Break}
Bringmann, Gr\o nlund, and Larsen~\cite{bringmann2017dichotomy} studied the decision version of the word break problem, and gave an algorithm with $\tilde{O}(nm^{1/3}+m)$ running time (with a matching conditional lower bound for combinatorial algorithms).

We consider the optimization version of the problem (with unit weight), defined as follows:

\begin{problem} {\sc (Minimum Word Break)}
Given a text string $s$ with length $n$ and a dictionary $D=\{d_1,\dots,d_k\}$ with total length $m$, find the minimum number $t^*$ such that $s$ can be split into $t^*$ words in $D$ (duplicates are allowed).
\end{problem}

The single-target change-making problem can be viewed as a special case of this problem, by representing each coin with value $v_i$ as a string with length $v_i$ over a unary alphabet.


In this section, we show that Bringmann et al.'s algorithm can be modified to solve the minimum word break problem
without increasing the running time (ignoring polylogarithmic factors), by using our Lemma~\ref{lemma:1} for $(\min,+)$-convolution.



\paragraph{Previous algorithm.}
We begin with a sketch of Bringmann et al.'s previous algorithm, which actually
solves an extension of the problem: compute a Boolean array $S[1,\ldots,n]$, where
$S[i]=1$ iff the prefix $s[1..i]$ can be broken into words in $D$.

For each $q\leq n$ being a power of $2$, let $D_q$ be the set of all strings in $D$ with length between $q$ and $2q-1$.  Bringmann et al.~\cite{bringmann2017dichotomy} introduced the following subproblem (which they called ``Jump Query''):

\begin{problem}\label{prob:jump}
Given $q\le n$ being a power of 2, an index $x\le n$, and a Boolean array $S[x-2q+1,\ldots,x]$,
compute a new Boolean array $S'[x+1,\ldots,x+q]$ where
$S'[i]=1$ iff there exists $j$ such that
$S[j-1]=1$ and $s[j..i]$ is a word in $D_q$.
\end{problem}

Bringmann et al.\ observed that the original problem reduces
to $O(n/q)$ instances of Problem~\ref{prob:jump} with parameter~$q$, over all $q$'s that are powers of 2.
(In the special case when all strings in $D$ have roughly the same length in $[q,2q)$,
i.e., $D=D_q$,
the observation is easy to see: we can generate the array $S[1,\ldots,n]$ from left to right, and solve an instance of Problem~\ref{prob:jump}
for every index $x$ divisible by~$q$.  In the general case, we run these processes for all $q$
simultaneously, and take the element-wise OR of the outputs, as we proceed from left to right.)

To solve Problem~\ref{prob:jump}, Bringmann et al.'s approach is to build a trie $\mathcal{T}_q$ for $D_q^{\mathrm{rev}}$, the reverse of the strings in $D_q$. The nodes in the trie $\mathcal{T}_q$ that spell the strings in $D_q^{\mathrm{rev}}$ are \emph{marked}.
(Here, a node $v$ \emph{spells} the string formed
by concatenating the symbols on the path from the root to~$v$.)
As a first step, we generate a \emph{maximal} collection $\mathcal{B}$ of node-disjoint downward paths in $\mathcal{T}_q$, satisfying the property that
each path $B\in \mathcal{B}$ contains exactly $\lambda_q$ marked nodes, where $\lambda_q$ is a parameter to be set later.
(The construction of $\mathcal{B}$ is simple
and involves just a depth-first search and a counter; see~\cite[Lemma~5]{bringmann2017dichotomy}.)
The size of $\mathcal{B}$ can be bounded by $|\mathcal{B}|\leq \frac{m}{q\cdot \lambda_q}$, since there are only $|D_q^{\mathrm{rev}}|\leq m/q$ marked nodes in the trie $\mathcal{T}_q$, and each path $B\in \mathcal{B}$ contains exactly $\lambda_q$ marked nodes. (Note that $|\mathcal{B}|=0$ if $\lambda_q>m/q$.)

To compute $S'[i]$ for a given $i\in\{x+1,\ldots,x+q\}$, we want to decide whether there exists an index $j^*$ such that
$S[j^*-1]=1$ and $s[j^*..i]$ is a word in $D_q$. To this end, we first find the node $v$ in $\mathcal{T}_q$ spelling the longest prefix of $s[1..i]^{\mathrm{rev}}$ that is in $D_q^{\mathrm{rev}}$.  (The node $v$ can be found quickly using suffix trees; see
\cite[Lemma~4]{bringmann2017dichotomy}.)
In order for $s[j^*..i]$ to be a word in $D_q$, the node $u^*$
spelling $s[j^*..i]^{\mathrm{rev}}$---which is the $(i-j^*+2)$-th node on the
path from the root to $v$---must be marked.
To search for $u^*$ (and thus $j^*$), starting from $v$, we repeatedly visit the next lowest marked ancestor in $\mathcal{T}_q$ (and check whether $S[j^*-1]=1$ for the corresponding $j^*$), until we reach the top marked node $r_B$ of some path $B\in \mathcal{B}$, or we reach the root. This takes at most $2\lambda_q$ steps by maximality of $\mathcal{B}$, and each step takes $O(1)$ time.

It remains to search for $u^*$ among all marked nodes on the path from $r_B$ to the root. Let $S_B$ be a Boolean array where $S_B[k]=1$ iff the $k$-th node on the path from the root to $r_B$ is marked. We precompute the Boolean convolution
between $S[x-2q+1,\ldots,x]$ and $S_B[q+1,\ldots,2q]$.
We want to decide the existence of an index $j^*$ with $S[j^*-1]=1$ and $S_B[i-j^*+2]=1$.
Thus, the answer $S'[i]$ can be determined by examining the entry $(S\circ S_B)[i+1]$ of the convolution.

The total cost of precomputing the above $|\mathcal{B}|$ Boolean convolutions on $O(q)$-sized arrays is
\[ O(|\mathcal{B}|\cdot q\log q) \ =\ \OO\left(\frac{m}{q\cdot \lambda_q} q\right) \ =\ \OO\left(\frac m{\lambda_q}\right).
\]
In addition, we spend $O(\lambda_q)$ time for each $i\in\{x+1,\ldots,x+q\}$; the total additional cost is $O(q\cdot \lambda_q)$.  The total time
is
\[ O\left(q\lambda_q + \frac m{\lambda_q}\right).
\]
To balance cost, we choose $\lambda_q=\sqrt{m/q}$, and as a result, Problem~\ref{prob:jump} can be solved in $\OO(\sqrt{mq})$ time.

Bringmann et al.\ also noted a more naive  $\OO(q^2)$-time
algorithm for Problem~\ref{prob:jump} (which we omit since we will not need it).
So the final
running time for the word break problem is
\[ \OO\left( \sum_{q=2^\ell,~\ell\leq \log n} \frac nq\cdot \min\{\sqrt{mq},\, q^2\} \right) \ =\  \OO(n m^{1/3})
\]
(as the largest term occurs when $q$ is near $m^{1/3}$),
plus $\OO(n+m)$ for preprocessing.

\paragraph{Modified algorithm and analysis.}
We now modify Bringmann et al.'s algorithm to solve the minimum word break problem.
Problem~\ref{prob:jump} is changed to the following subproblem:
\begin{problem}\label{prob:jump2}
Given $q\le n$ being a power of 2, an index $x\le n$, and an array of numbers $S[x-2q+1,\ldots,x]$,
compute a new array $S'[x+1,\ldots,x+q]$ where
$S'[i]=\min\{ S[j-1]+1: \mbox{$s[j..i]$ is a word in $D_q$}\}$.
\end{problem}

The minimum word break problem reduces to Problem~\ref{prob:jump2} like before (taking element-wise minimum instead of OR).
We solve Problem~\ref{prob:jump2} like before, except that we
take the $(\min,+)$-convolution $S\star S_B$ instead of Boolean convolution
(the entries of $S_B$ are now in $\{1,\infty\}$ instead of $\{0,1\}$).
For each index $i\in\{x+1,\dots,x+q\}$, we are interested in a specific entry $(S\star S_B)[i+1]$ for one specific path $B\in \mathcal{B}$. Equivalently, for each path $B\in \mathcal{B}$, we are only interested in $q_B$ entries in the output array $S\star S_B$,
for some $q_B$'s with $\sum_{B\in \mathcal{B}}q_B=q$.  We use the output-sensitive part of Lemma~\ref{lemma:1} to compute these $(\min,+)$-convolutions on $O(q)$-sized arrays (since $S_B$ is binary).  Thus, we can perform the $(\min,+)$-convolution for a path $B$
in $\OO(q\sqrt{q_B})$ time.
By the Cauchy--Schwarz inequality, the sum of the cost over all paths $B\in\mathcal{B}$ is
\[ \tilde{O}\left(\sum_{B\in \mathcal{B}}q\sqrt{q_B}\right)\,=\,
\tilde{O}\left(
q\sqrt{q|\cal B|}
\right)\,=\,
\tilde{O}\left(q\sqrt{\frac{m}{\lambda_q}}\right).
\]
In addition, we spend $O(\lambda_q)$ time for each $i\in\{x+1,\ldots,x+q\}$; the total additional cost is $O(q\cdot \lambda_q)$.
The total time is
\[ \tilde{O}\left(q \lambda_q \,+\, q\sqrt{\frac{m}{\lambda_q}}\right).
\]
To balance cost, we choose $\lambda_q=m^{1/3}$, and as a result, Problem~\ref{prob:jump2} can be solved in $\tilde{O}(qm^{1/3})$ time.

So the final running time for the minimum word break problem is
%
$$\tilde{O}\left(\sum_{q=2^\ell,~\ell\leq \log n}\frac{n}{q}\cdot qm^{1/3}\right)
\,=\, \OO(nm^{1/3}),$$
plus $\OO(n+m)$ for preprocessing, which luckily gives the same result as Bringmann et al.'s previous algorithm.


\begin{theorem}\label{lemma:12}
The minimum word break problem can be solved in $\tilde{O}(n m^{1/3} + m)$ time.
\end{theorem}


\paragraph{Remark.}
Note that the algorithm actually solves an extension of the problem: compute the minimum number of breaks for every prefix of the input string.
In particular, when the alphabet is unary, this implies
an
 $\tilde{O}(t\sigma^{1/3}+\sigma)$-time
 algorithm for
the all-targets change-making problem.  However, this bound is not as good as those from Theorem~\ref{thm:1} and Corollary~\ref{cor:2}  ($\tilde{O}( \min\{t^{4/3},\,(t\sigma)^{2/3}+t\})$).

\section{Concluding Remarks}

Our change-making algorithms can be modified to compute not just the minimum number of coins but also a representation of the minimum multiset of coins for every target value. 
For the FFT-based algorithms, we need standard techniques for witness finding~\cite{AlonGMN92,Seidel92} (which only increases the running time by polylogarithmic factors).

Although Erd\H{o}s and Graham's $\Theta(u^2/k)$ bound on the Frobenius problem is asymptotically tight in the worst case
(one bad coin set is $\{x,2x,\ldots,(k-1)x,(k-1)x-1\}$ with $x=\lceil\frac{u}{k-1}\rceil$), the Frobenius number tends to be smaller
for ``many'' $k$-tuples of coin values (it is usually subquadratic even for $k=3$).  This suggests that
our $\OO(u^2+t)$-time algorithm for all-targets coin changing might be improvable for many input sets of coins.
However, obtaining an improvement in the worst case remains intriguingly open (this might require new results on the Frobenius problem---the interplay between combinatorial and algorithmic results seems worthy of further study).



\paragraph{Acknowledgement.} We thank Adam Polak and Chao Xu for discussion and, in particular, for bringing the minimum word break problem to our attention.  We also thank the anonymous reviewers for their helpful comments.

{\small
\bibliographystyle{plain}
\bibliography{references}
}

\appendix

\end{document}